%% file: arxiv_version.tex
\documentclass[11pt]{article}
\usepackage[margin=1in]{geometry}
\usepackage{graphicx}
\usepackage{amsmath}
\usepackage{amssymb}
\usepackage{amsthm}
\usepackage{xcolor}
\usepackage{microtype}
\usepackage[ruled,noend]{algorithm2e}
\usepackage[backend=biber,isbn=false,date=year,maxnames=99]{biblatex}
\usepackage[colorlinks=true]{hyperref}
\hypersetup{
    linkcolor = {violet},
    citecolor = {violet},
    urlcolor = {teal}
}

\newenvironment{acks}{\section*{Acknowledgments}}{}

\bibliography{ziptrees}

\theoremstyle{plain}
\newtheorem{theorem}{Theorem}
\newtheorem{lemma}{Lemma}

\theoremstyle{definition}

\input{macros}

\title{Zip Trees}

\author{
  Robert E. Tarjan\footnote{Department of Computer Science, Princeton University, and Intertrust Technologies; \mbox{ret@cs.princeton.edu}.}\\
  \and
  Caleb C. Levy\footnote{Sunshine; \mbox{caleb.levy@gmail.com}.}\\
  \and
  Stephen Timmel\footnote{Mathematics Department, Virginia Polytechnic Institute and State University; \mbox{stimmel@vt.edu}.}\\
}

\date{}

\begin{document}
\maketitle
\input{text/abstract}
\input{sections}

\printbibliography

\end{document}

%% file: macros.tex
\let\oldnl\nl%
\newcommand{\nonl}{\renewcommand{\nl}{\let\nl\oldnl}}%

\newcommand\Left{\mathit{left}}
\newcommand\Right{\mathit{right}}
\newcommand\Key{\mathit{key}}
\newcommand\Root{\mathit{root}}
\newcommand\Rank{\mathit{rank}}
\newcommand\Null{\mathit{null}}

%% file: text/abstract.tex
\begin{abstract}
We introduce the \emph{zip tree},\footnote{\emph{Zip}: ``To move very fast.''}
a form of randomized binary search tree that integrates previous ideas into one
practical, performant, and pleasant-to-implement package. A zip tree is a
binary search tree in which each node has a numeric rank and the tree is
(max)-heap-ordered with respect to ranks, with rank ties broken in favor of
smaller keys. Zip trees are essentially treaps \cite{Aragon}, except that ranks
are drawn from a geometric distribution instead of a uniform distribution, and
we allow rank ties. These changes enable us to use fewer random bits per node.

We perform insertions and deletions by unmerging and merging paths
(\emph{unzipping} and \emph{zipping}) rather than by doing rotations, which
avoids some pointer changes and improves efficiency. The methods of zipping and
unzipping take inspiration from previous top-down approaches to insertion and
deletion by Stephenson \cite{Stephenson}, Mart\'inez and Roura \cite{Martinez},
and Sprugnoli \cite{Sprugnoli}.

From a \emph{theoretical} standpoint, this work provides two main
results. First, zip trees require only $O(\log \log n)$ bits (with high
probability) to represent the largest rank in an $n$-node binary search tree;
previous data structures require $O(\log n)$ bits for the largest rank. Second,
zip trees are naturally isomorphic to skip lists \cite{Pugh}, and simplify Dean
and Jones' mapping between skip lists and binary search trees \cite{Dean}.
\end{abstract}

%% file: sections.tex
\input{text/introduction}
\input{text/relatedwork}
\input{text/properties}
\input{text/extensions}
\input{text/implementation}
\input{text/acknowledgements}

%% file: text/introduction.tex
\section{Introducing: Zip Trees}\label{sec:Introduction}

\subsection{Preliminaries}

A \emph{binary search tree} is a binary tree in which each node contains an
item, each item has a key, and the items are arranged in \emph{symmetric
order}: if $x$ is a node, all items in the left subtree of $x$ have keys less
than that of $x$, and all items in the right subtree of $x$ have keys greater
than that of $x$. Such a tree supports binary search: to find an item in the
tree with a given key, proceed as follows. If the tree is empty, stop: no item
in the tree has the given key. Otherwise, compare the desired key with that of
the item in the root. If they are equal, stop and return the item in the root.
If the given key is less than that of the item in the root, search recursively
in the left subtree of the root. Otherwise, search recursively in the right
subtree of the root. The path of nodes visited during the search is the
\emph{search path}. If the search is unsuccessful, the search path starts at
the root and ends at a missing node corresponding to an empty subtree.

To keep our presentation simple, in this and the next section we do not
distinguish between an item and the node containing it. (The data structure is
\emph{endogenous} \cite{Tarjan}.) We also assume that all nodes have distinct
keys. It is straightforward to eliminate these assumptions. We call a node
\emph{binary}, \emph{unary}, or a \emph{leaf}, if it has two, one or zero
children, respectively. We define the \emph{depth} of a node recursively to be
zero if it is the root, or one plus the depth of its parent if not. We define
the \emph{height} of a node recursively to be zero if it is a leaf, or one plus
the maximum of the heights of its children if not. The \emph{left} (resp.
\emph{right}) \emph{spine} of a node $x$ is the path from $x$ through left (resp.
\emph{right}) children to the node of smallest (resp. largest) key in the subtree rooted at
$x$. The \emph{left} (resp. \emph{right}) \emph{spine} of a tree is the left
(resp. right) spine of its root. We represent a binary search tree by storing
in each node $x$ its left child $x.\Left$, its right child $x.\Right$, and its
key, $x.\Key$. If $x$ has no left (resp. right) child, $x.\Left = \Null$ (resp.
$x.\Right = \Null$).

\subsection{Intuition}

Our goal is to obtain a type of binary search tree with small depth and small
update time, one that is as simple and efficient as possible. If the number of
nodes $n$ is one less than a power of two, the binary tree of minimum depth is
\emph{perfect}: each node is either binary (with two children) or a leaf (with
no children), and all leaves are at the same depth. But such trees exist only
for some values of $n$, and updating even an almost-perfect tree (say one in
which all non-binary nodes are leaves and all leaves have the same depth to
within one) can require rebuilding much or all of it.

We observe, though, that in a perfect binary tree the fraction of nodes of
height $k$ is about $1/2^{k+1}$ for any non-negative integer $k$. Our idea is
to build a good tree by assigning heights to new nodes according to the
distribution in a perfect tree and inserting the nodes at the corresponding
heights.

We cannot do this exactly, but we can do it to within a constant factor in
expectation, by assigning each node a random rank according to the desired
distribution and maintaining heap order by rank. Thus we obtain zip trees.

\subsection{Definition of Zip Trees}

A \emph{zip tree} is a binary search tree in which each node has a numeric
\emph{rank} and the tree is (max)-heap-ordered with respect to ranks, with ties
broken in favor of smaller keys: the parent of a node has rank greater than
that of its left child and no less than that of its right child. We choose the
rank of a node randomly when the node is inserted into the tree. We choose node
ranks independently from a geometric distribution with mean $1$: the rank of a
node is non-negative integer $k$ with probability $1/2^{k+1}$. We denote by
$x.\Rank$ the rank of node $x$. We can store the rank of a node in the node or
compute it as a pseudo-random function of the node (or of its key) each time it
is needed. The pseudo-random function method, proposed by Aragon and Seidel
\cite{Aragon}, avoids the need to store ranks but requires a stronger
independence assumption for the validity of our efficiency bounds, as we
discuss in Section \ref{sec:Properties}.

To insert a new node $x$ into a zip tree, we search for $x$ in the tree until
reaching the node $y$ that $x$ will replace, namely the node $y$ such that
$y.\Rank \le x.\Rank$, with strict inequality if $y.\Key < x.\Key$. From $y$,
we follow the rest of the search path for $x$, \emph{unzipping} it by splitting
it into a path $P$ containing each node with key less than $x.\Key$ and a path
$Q$ containing each node with key greater than $x.\Key$. Along $P$ from top to
bottom, nodes are in increasing order by key and non-increasing order by rank;
along $Q$ from top to bottom, nodes are in decreasing order by both rank and
key. Unzipping preserves the left subtrees of the nodes on $P$ and the right
subtrees of the nodes on $Q$. We make the top node of $P$ the left child of $x$
and the top node of $Q$ the right child of $x$. Finally, if $y$ had a parent
$z$ before the insertion, we make $x$ the left or right child of $z$ depending
on whether its key is less than or greater than that of $z$, respectively ($x$
replaces $y$ as a child of $z$); if $y$ was the root before the insertion, we
make $x$ the root.

Deletion is the inverse of insertion. To delete a node $x$, we do a search to
find it. Let $P$ and $Q$ be the right spine of the left subtree of $x$ and the
left spine of the right subtree of $x$. \emph{Zip} $P$ and $Q$ to form a single
path $R$ by merging them from top to bottom in non-increasing rank order,
breaking a tie in favor of the smaller key. Zipping preserves the left subtrees
of the nodes on $P$ and the right subtrees of the nodes on $Q$. Finally, if $x$
had a parent $z$ before the insertion, make the top node of $R$ (or $\Null$ if
$R$ is empty) the left or right child of $z$, depending on whether the key of
$x$ is less than or greater than that of $z$, respectively (the top node of $R$
replaces $x$ as a child of $z$); if $x$ was the root before the insertion, make
the top node of $R$ the root. Figure \ref{fig:InsertDelete} illustrates both an
insertion and a deletion in a zip tree.

\begin{figure}
\centering
\includegraphics[width=\textwidth]{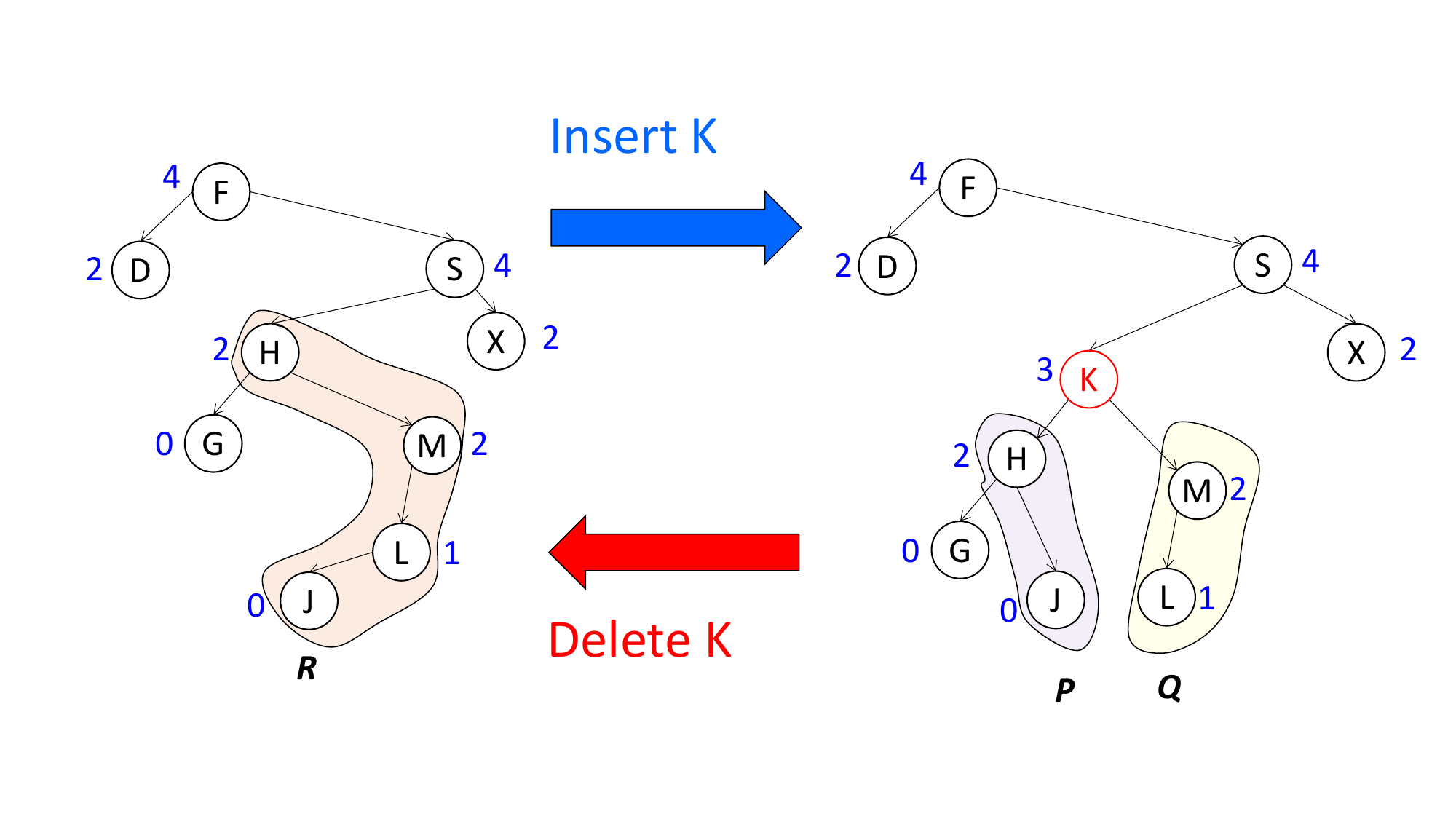}
\caption{Insertion and deletion of a node with key ``K'' assigned rank
$3$.}\label{fig:InsertDelete}
\end{figure}

An insertion or deletion requires a search plus an unzip or zip. The time for
an unzip or zip is proportional to one plus the number of nodes on the unzipped
path in an insertion or one plus the number of nodes on the two zipped paths in
a deletion.

%% file: text/relatedwork.tex
\section{Related Work}\label{sec:PreviousWork}

Zip trees closely resemble two well-known data structures: the treap of Seidel
and Aragon \cite{Aragon} and the skip list of Pugh \cite{Pugh}.

\subsection{Treaps}

A \emph{treap} is a binary search tree in which each node has a real-valued
random rank (called a \emph{priority} by Seidel and Aragon) and the nodes are
max-heap ordered by rank. The ranks are chosen independently for each node from
a fixed, uniform distribution over a large enough set that the probability of
any rank tie is small. Insertions and deletions are done using \emph{rotations}
to restore heap order. A rotation at a node $x$ is a local transformation that
makes $x$ the parent of its old parent while preserving symmetric order. See
Figure~\ref{fig:Rotation}. In general a rotation changes three children.

\begin{figure}
\centering
\includegraphics[width=0.8\columnwidth,trim={.15in 1.3in .25in
2in},clip]{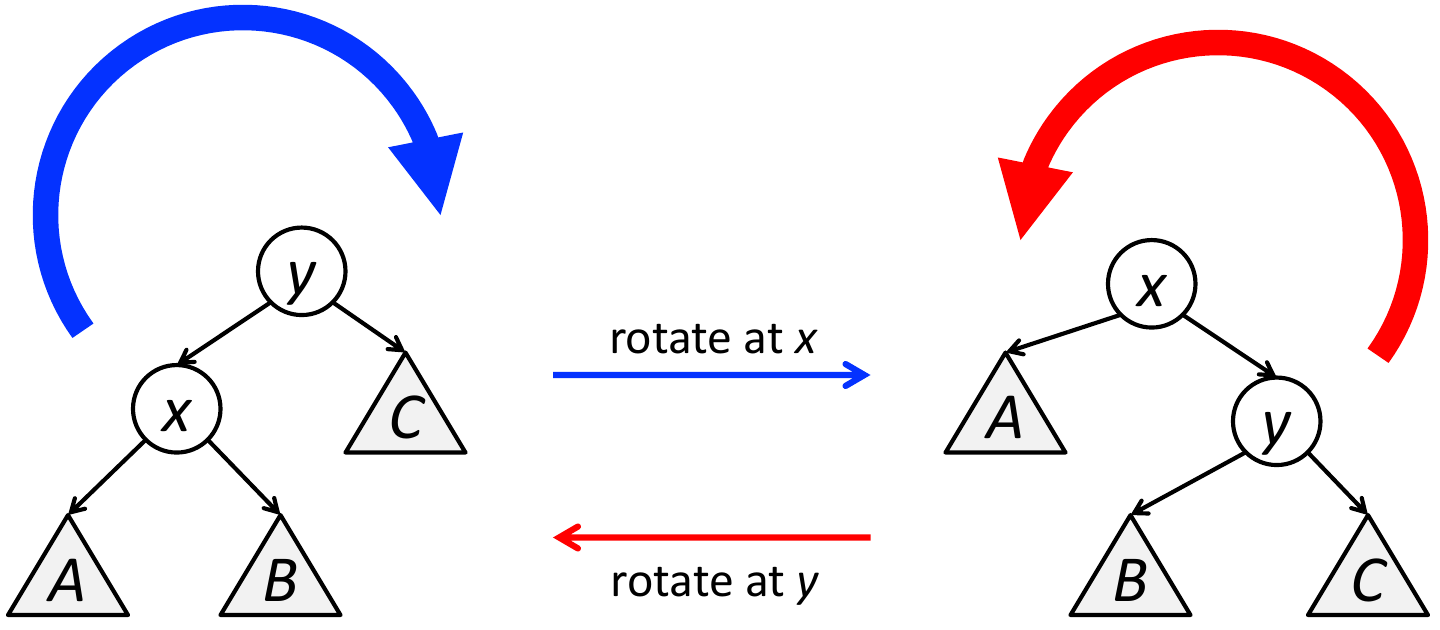}
\caption{Rotation at node $x$ with parent $y$, and reversing the effect by
rotating at $y$. Triangles denote subtrees.}\label{fig:Rotation}
\end{figure}

To insert a new node $x$ in a treap, we generate a rank for $x$, follow the
search path for $x$ until reaching a missing node, replace the missing node by
$x$, and rotate at $x$ until its parent has larger rank or $x$ is the root. To
delete a node $x$ in a treap, while $x$ is not a leaf, we rotate at whichever
of its children has higher rank. %
Once $x$ is a leaf or a unary node, we replace it by its child if it
has one or by a missing node if not.

One can view a zip tree as a treap but with a different choice of ranks and
with different insertion and deletion algorithms.\footnote{Seidel and Aragon
\cite{Aragon} hint at the possibility of doing insertions and deletions by
unzipping and zipping: in a footnote they say, ``In practice it is preferable
to approach these operations the other way around. Joins and splits of treaps
can be implemented as iterative top-down procedures; insertions and deletions
can then be implemented as accesses followed by splits or joins.'' But they
provide no further details.} Our choice of ranks reduces the number of bits
needed to represent them from $O(\log n)$ to $\lg\lg n + O(1)$ (Theorem
\ref{thm:ExpectedRank}), if ranks are stored rather than computed as a function
of the node or its key. Treaps have the same expected depth as search trees
built by uniformly random insertions, namely $2\ln n$, about $1.39\lg n$, as
compared to $1.5\lg n$ for zip trees. The results in Section
\ref{sec:Properties} correspond to results for treaps. Allowing rank ties as we
do thus costs about 8\% in average depth (and search time) but allows much more
compact representation of priorities.

A precursor of the treap is the \emph{cartesian tree} of Jean Vuillemin
\cite{Vuillemin}. This is a binary search tree built by leaf insertion (search
for the item; insert it where the search exits the bottom of the tree), with
each node having a priority equal to its position in the sequence of
insertions. Such a tree is min-heap ordered with respect to priorities, and its
distributional properties are the same as those of a treap if items are
inserted in an order corresponding to a uniformly random permutation.

\subsection{Skip Lists}

A \emph{skip list} is an alternative randomized data structure that supports
logarithmic comparison-based search. It consists of a hierarchy of sublists of
the items. The level-$0$ list contains all the items. For $k > 0$, the
level-$k$ list is obtained by independently adding each item of the
level-$(k-1)$ list with probability $1/2$ (or, more generally, some fixed
probability $p$). Each list is in increasing order by key. A dummy item with
key less than those of all real items is added to each list.

A search starts in the top-level list and proceeds through the items in
increasing order by key until finding the desired item, reaching an item of
larger key, or reaching the end of the list. In either of the last two cases,
the search backs up to the item of largest key less than the search key,
descends to the copy of this item in the next lower-level list, and searches in
this list in the same way. Eventually the search either finds the item or
discovers that it is not in the level-$0$ list.

\begin{figure}
\centering
\includegraphics[width=\textwidth]{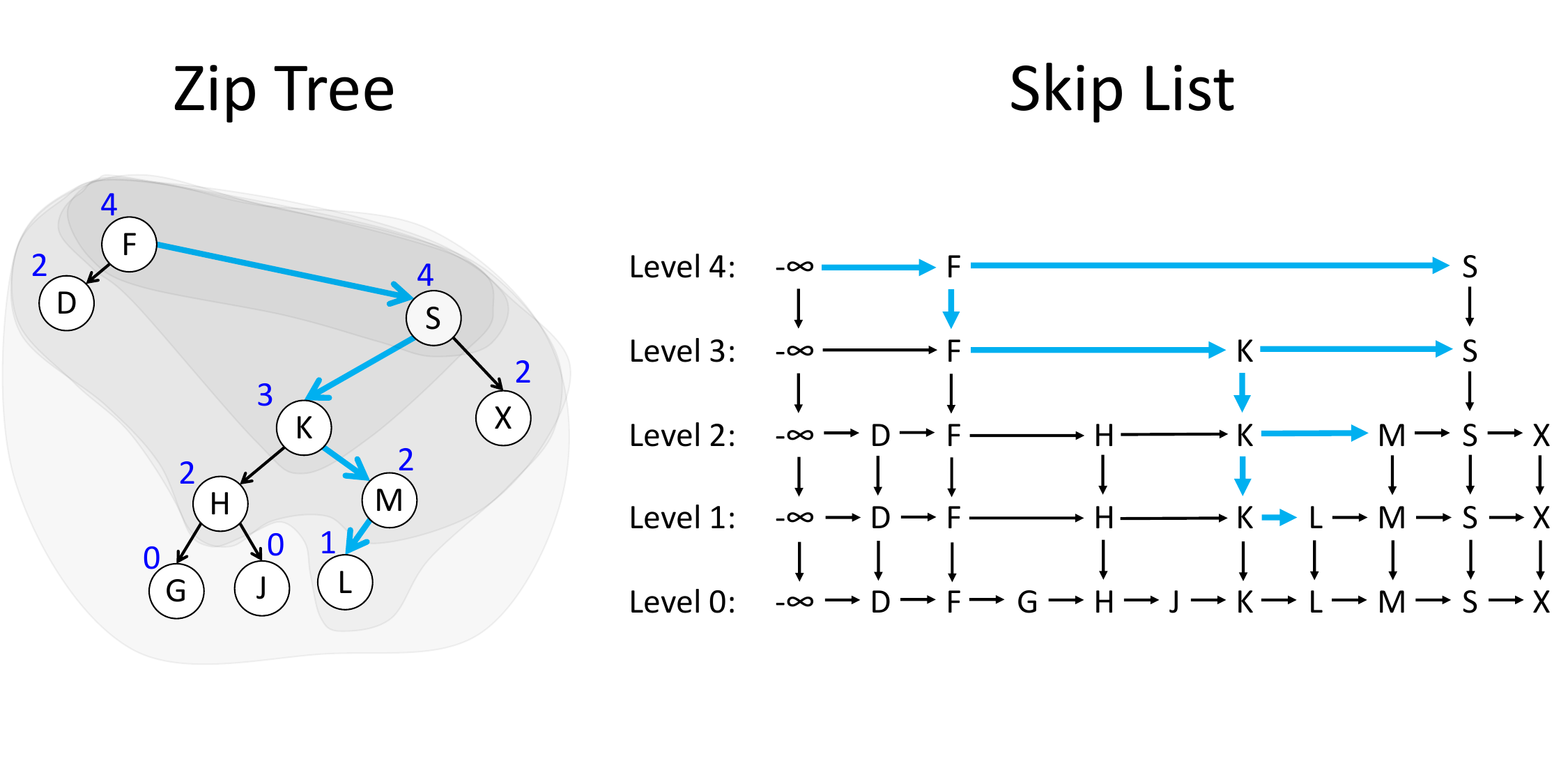}
\caption{Representation of the zip tree in Figure \ref{fig:InsertDelete} as a
skip list. The level-$k$ sublist comprises nodes of the zip tree of rank $k$ or
less. The search for $L$ traverses the blue arcs in the zip tree and in the
corresponding skip list.}\label{fig:SkipListZipTree}
\end{figure}

One can view a zip tree as a compact representation of a skip list. There is a
natural isomorphism between zip trees and skip lists. (See Figure
\ref{fig:SkipListZipTree}.) Given a zip tree, the isomorphic skip list contains
item $e$ in the level-$k$ sublist if and only if $e$ has rank at least $k$ in
the zip tree. Given a skip list, the isomorphic zip tree contains item $e$ with
rank $k$ if and only if $e$ is in the level-$k$ sublist but not in the
level-$(k+1)$ sublist. Let $e$ be an item in the zip tree with left and right
children $e'$ and $e''$, respectively. Let $e$, $e'$, and $e''$ have ranks $k$,
$k'$, and $k''$, respectively. A search in the skip list that reaches an
occurrence of $e$ will reach it first in the level-$k$ sublist. The next node
visited during the search that is not an occurrence of $e$ will be the
occurrence of $e'$ in the level-$k'$ sublist or the occurrence of $e''$ in the
level-$k''$ sublist, depending on whether the search key is less than or
greater than the key of $e$. Our rule for breaking rank ties in zip trees is
based on the search direction in skip lists: from smaller to larger keys.

A search in a zip tree visits the same items as the search in the isomorphic
skip list, except that the latter may visit items repeatedly, at lower and
lower levels. Thus a zip tree search is no slower than the isomorphic skip list
search, and can be faster. The skip list has at least as many pointers as the
corresponding zip tree, and its representation requires either variable-size
nodes, in which each item of rank $k$ has a node containing $k+1$ pointers; or
large nodes, all of which are able to hold a number of pointers equal to the
maximum rank plus one; or small nodes, one per item per level, requiring
additional pointers between levels. We conclude that zip trees are at least as
efficient in both time and space as skip lists.

Dean and Jones were the first to provide a mapping that converts a skip list
into a binary search tree \cite{Dean}, but it is not the natural isomorphism
given in the previous paragraph. They store ranks in the binary search tree in
difference form. They map the insertion and deletion algorithms for a skip list
into algorithms on the corresponding binary search tree by using rotations.

\subsection{Other Binary Search Tree Algorithms}

Mart\'inez and Roura \cite{Martinez} proposed insertion and deletion
algorithms that produce trees with the same distribution as treaps. Instead of
maintaining a heap order with respect to random priorities, they do insertions
and deletions via random rotations that depend on subtree sizes. These sizes
must be stored, at a cost of $O(\log n)$ bits per node, and they must be
updated after each rotation. This suggests using their method only in an
application in which subtree sizes are needed for some other purpose.

Doing insertions and deletions via unzipping and zipping takes at most one
child change per node on the restructured path or paths, saving a constant
factor of at least three over using rotations. Stephenson used unzipping in his
root insertion algorithm \cite{Stephenson}; insertion by unzipping is a hybrid
of his algorithm and leaf insertion. Sprugnoli \cite{Sprugnoli} was the first
to propose insertion by unzipping. He used it to insert a new node at a
specified depth, with the depth chosen randomly. His proposals for the depth
distribution are complicated, however, and he did not consider the possibility
of choosing an approximate depth rather than an exact depth. Zip trees choose
the insertion \emph{height} approximately rather than the depth, a crucial
difference.

%% file: text/properties.tex
\section{Properties of Zip Trees}\label{sec:Properties}

If we ignore constant factors, the properties of zip trees are shared with
treaps, skip lists, and the randomized search trees of Mart\'inez and Roura.
Because zip trees use a different rank distribution than treaps and Mart\'inez
and Roura's trees, and a different representation than skip lists, we reprove
a selection of these properties.

\subsection{Behavior of Node Ranks}

We denote by $n$ the number of nodes in a zip tree. To simplify bounds, we
assume that $n>1$, so $\log n$ is positive. We denote by $\lg n$ the base-two
logarithm. The following lemma extends a well-known result for trees
symmetrically ordered by key and heap-ordered by rank \cite{Aragon} to
allow rank ties:

\begin{lemma}[From \cite{Aragon}]
\label{lem:UniqueTree}
The structure of a zip tree is uniquely determined by the keys and ranks of its
nodes.
\end{lemma}
\begin{proof}
The lemma is immediate by induction on $n$, since the root is the node of
largest rank whose key is smallest, and the nodes in the left and right
subtrees of the root are those with keys less than and greater than the key of
the root, respectively.
\end{proof}

By Lemma \ref{lem:UniqueTree}, a zip tree is \emph{history-independent}: its
structure depends only on the nodes it currently contains (and their ranks),
independent of the sequence of insertions and deletions that built it.

In our efficiency analysis we assume that each deletion depends only on the
sequence of previous insertions and deletions, independent of the node ranks.
(If an adversary can choose deletions based on node ranks, it is easy to build
a bad tree: insert items in arbitrary order; if any item has a rank greater
than $0$, immediately delete it. This will produce a path containing half the
inserted nodes on average.)

\begin{theorem}
\label{thm:ExpectedRank}
The expected rank of the root in a zip tree is at most $\lg n +3$. For any
$c>0$, the root rank is at most $(c+1)\lg n$ with probability at most $1-1/n^c$.
\end{theorem}

\begin{proof}
The root rank is the maximum of $n$ samples of the geometric distribution with
mean $1$. For $c > 0$, the probability that the root rank is at least $\lg n +
c$ is at most $n/2^{\lg n + c} = 1/2^c$. It follows that the expected root rank
is at most $\lceil \lg n \rceil + \sum_{i=1}^\infty i/2^i \le \lceil \lg n
\rceil +2 \le \lg n + 3$. For $c > 0$, the probability that the root rank
exceeds $(c + 1)\lg n$ is at most $1/2^{c \lg n} = 1/n^c$.
\end{proof}

Let $x$ be a node in a zip tree. If $y$ is on the search path for $x$ then $y$
is an \emph{ancestor} of $x$ and $x$ is a \emph{descendant} of $y$. The
\emph{low} (respectively \emph{high}) ancestors of $x$ are the ancestors of $x$
with key less than (respectively greater than) that of $x$.

\begin{lemma}
\label{lem:LowAncestors}
Let $x$ be a node and let $\ell$ be the number of low ancestors of $x$. Then
the expected value of $\ell$ is at most $k+1$, and for any $\delta > 0$, $\ell
\leq (1+\delta)k+1$ with probability at least $1-e^{-\frac{\delta^2
k}{2+\delta}}$.
\end{lemma}

\begin{proof}
If we order the low ancestors of $x$ in increasing order by key, they are in
non-increasing order by rank. We can think of these ancestors and their ranks
as being generated by coin flips in the following way. At each successive node
$y$ less than $x$ in decreasing key order we flip a fair coin until it comes up
tails and give $y$ a rank equal to the number of heads. Given such a $y$, let
$z$ be the low ancestor of smallest key greater than that of $y$ if there is
such a low ancestor; otherwise, let $z=x$. Then $y$ is a low ancestor of $x$ if
and only if its rank is at least the rank of $z$. We call the first $z.\Rank$
coin flips at $y$ \emph{irrelevant} and the rest \emph{relevant}.

For example, consider the tree on the left in Figure \ref{fig:InsertDelete},
and refer to nodes by their keys. If $x$ is node J, the successive nodes less
than $x$ in decreasing order are H, G, F, D. Initially $y$ is H and $z$ is J.
The flips at H are head, head, tail, giving H a rank of $2 \geq 0$, so H is a
low ancestor of J and is the next value of $z$. All three flips at H are
relevant. The next value of $y$ is G, for which the first flip is a tail,
giving G a rank of 0. Since $0 < 2$, G is not a low ancestor of J; $z$ stays
equal to H. The next value of $y$ is F, for which the flips are four heads
followed by a tail, making F a low ancestor of J and changing $z$ to F. The
first two flips are irrelevant and the last three are relevant. The last value
of $y$ is D, for which the flips are two heads followed by a tail, giving D a
rank of $2$. Since $2 < 4$, D is not a low ancestor of J.

Node $y$ is a low ancestor of $x$ if and only if at least one flip at $y$ is
relevant. The relevant flips are a sequence of Bernoulli trials in which the
number of tails is the number of low ancestors of $x$ produced so far and the
number of heads is at most the rank of the low ancestor of $x$ of highest rank
produced so far. Thus $\ell$ is the number of tails in a sequence of flips
containing at most $k$ heads and ending with a tail. Since the expected number
of tails preceding each head is $1$, $\ell \leq k +1$ in expectation: the
``$+1$'' counts the last tail. The second half of the lemma follows by a
Chernoff bound \cite{Chernoff}.
\end{proof}

\begin{lemma}
\label{lem:HighAncestors}
Let $x$ be a node and let $h$ be the number of high ancestors of $x$. Then the
expected value of $h$ is at most $k/2$, and for any $\delta > 0$, $h \leq
(1+\delta)k/2$ with probability at least $1-e^{-\frac{\delta^2
k}{2(2+\delta)}}$.
\end{lemma}

\begin{proof}
The proof is like that of Lemma \ref{lem:LowAncestors}. We think of generating
the high ancestors of x and their ranks by flipping a fair coin until it comes
up tails at each node $y$ greater than $x$ in increasing key order and giving
$y$ a rank equal to the number of heads. Given such a $y$, let $z$ be the high
ancestor of $x$ of largest key smaller than that of $y$, or $x$ if there is no
such high ancestor. We call the first $z.\Rank+1$ flips at $y$ irrelevant and
the rest relevant.

Node $y$ is a high ancestor of $x$ if and only if at least one flip at $y$ is
relevant. The relevant flips are a sequence of Bernoulli trials in which the
number of tails is the number of high ancestors of $x$ produced so far and the
number of flips is at most the rank of the high ancestor of $x$ of highest rank
produced so far. Thus, $h$ is the number of tails in a sequence of at most $k$
flips. This is at most $k/2$ in expectation. The second half of the lemma
follows by a Chernoff bound.
\end{proof}

\subsection{Expected Zip Tree Structure}

\begin{theorem}\label{thm:ExpectedDepth}
The expected depth of a node in a zip tree is at most $(3/2)\lg n + O(1)$. For
$c \ge 1$, the depth of a zip tree is $O(c \lg n)$ with probability at least $1
- 1/n^c$, where the constant inside the big ``$O$'' is independent of $n$ and
$c$.
\end{theorem}

\begin{proof}
One can prove this theorem using results from \cite{Prodinger}, but for
completeness we prove it from scratch. The expected rank of the root is at most
$\lg n + 3$ by Theorem \ref{thm:ExpectedRank}. By Lemmas \ref{lem:LowAncestors}
and \ref{lem:HighAncestors}, the expected number of ancestors of a node $x$,
including $x$, is at most $(3/2)\lg n + O(1)$. The second half of the Theorem
follows from the high-probability bounds in Theorem \ref{thm:ExpectedRank} and
Lemmas \ref{lem:LowAncestors} and \ref{lem:HighAncestors}.
\end{proof}

By Theorem \ref{thm:ExpectedDepth}, the expected number of nodes visited during
a search in a zip tree is at most $(3/2)\lg n + O(1)$, and the search time is
$O(\log n)$ with high probability.

\begin{theorem}
\label{thm:ExpectedPathLengths}
If $x$ is a node of rank at most $k$, the expected number of nodes on the path
that is unzipped during its insertion, and on the two paths that are zipped
during its deletion, is at most $(3/2)k+2$. For any $\delta>0$, this number
is at most $(1+\delta)(3/2)k+2$ with probability at least
$1-2e^{-\frac{\delta^2 k}{2(2+\delta)}}$.
\end{theorem}

\begin{proof}
Let $x$ be a node of rank at most $k$. If $x$ is not in the tree but is
inserted, the nodes on the path unzipped during its insertion are exactly those
on the two paths that would be zipped during its deletion. Thus, we need only
consider deletion. Let $P$ and $Q$ be the two paths zipped during the deletion
of $x$, with $P$ containing the nodes of smaller key and $Q$ containing the
nodes of larger key. Let $y$ and $z$ be the predecessor and successor of $x$ in
key order, respectively. Then the nodes on $P$ are $y$ and the low ancestors of
$y$ of rank less than $x.\Rank$, and the nodes on $Q$ are $z$ and the high
ancestors of $z$ of rank at most $x.\Rank$. The theorem follows from Lemmas
\ref{lem:LowAncestors} and \ref{lem:HighAncestors}.
\end{proof}

\begin{theorem}
\label{thm:PointerChanges}
The expected number of pointer changes during an unzip or zip is $O(1)$. The
probability that an unzip or zip changes more than $k+O(1)$ pointers is at most
$2/c^k$ for some $c > 1$ independent of $k$.
\end{theorem}

\begin{proof}
The expected number of pointer changes is at most one plus the number of nodes
on the unzipped path during an insertion or the two zipped paths during a
deletion. For a given node $x$, these numbers are the same whether $x$ is
inserted or deleted. Thus we need only consider the case of deletion. The
probability that $x$ has rank $k$ is $1/2^{k + 1}$. Given that $x$ has rank
$k$, the expected number of nodes on the two zipped paths is at most $(3/2)k+2$ by Theorem \ref{thm:ExpectedPathLengths}. Summing over all possible
values of $k$ gives the first half of the theorem.

Choosing $\delta = 1$ in the second half of Theorem
\ref{thm:ExpectedPathLengths}, we find that if $x$ is a node of rank at most
$k/3$, the number of parent changes during its insertion or deletion is more
than $k + O(1)$ with probability at most $2/e^{k/18}$. Thus $c = e^{1/18}$
satisfies the second half of the theorem.
\end{proof}

By Theorem \ref{thm:PointerChanges}, the expected time to unzip or zip is
$O(1)$, and the probability that an unzip or zip takes $k$ steps is
exponentially small in $k$.

In some applications of search trees, each node contains a secondary data
structure, and making any change to a subtree may require rebuilding the entire
subtree, in time linear in the number of nodes. The following result implies
that zip trees are efficient in such applications.

\begin{theorem}\label{thm:Descendents}
The expected number of descendants of a node of rank $k$ is at most $3(2^k)-1$.
The expected number of descendants of an arbitrary node is at most $(3/2)\lg n
+ 2$.
\end{theorem}
\begin{proof}
Let $x$ be a node of rank $k$. Consider the nodes with key less than that of
$x$. Think of generating their ranks in decreasing order by key. The first such
node that is not a descendant of $x$ is the first one whose rank is at least
$k$. All nodes smaller than this one are also non-descendants of $x$. The
probability that a given node has rank at least $k$ is $1/2^k$. The probability
that the $i\textsuperscript{th}$ node is the first of rank at least $k$ is $(1
- 1/2^k)^{i - 1}/2^k$. The expected value of $i$ is the expected number of
flips of a biased coin to get a head, if the probability of a head is $p =
1/2^k$. By a standard calculation this number is $1/p = 2^k$. Thus the expected
number of descendants of $x$ of smaller key is at most $2^k-1$. (The expected
value of $i$ minus one is an overestimate because there are at most $n - 1$
nodes of key less than that of $x$ and they may all have smaller rank.)

Similarly, among the nodes with key greater than that of $x$, the first one
that is not a descendant of $x$ is the first one with rank greater than $k$. A
given node has rank greater than $k$ with probability $p = 1/2^{k + 1}$. The
probability that the $i\textsuperscript{th}$ node is the first of rank greater
than $k$ is $(1 - 1/2^{k + 1})^{i - 1}/2^{k + 1}$. The expected value of $i$ is
$2^{k + 1}$, so the expected number of descendants of $x$ of larger key is at
most $2^{k + 1} - 1$.

We conclude that the expected number of descendants of $x$, including $x$
itself, is at most $3(2^k) - 1$. The expected number of descendants of an
arbitrary node is the sum over all $k$ of the probability that the node has
rank $k$ times the expected number of descendants of the node given that its
rank is $k$. Using the fact that the number of descendants is at most $n$, this
sum is at most $(3/2)\lg n + 2$, since the terms for $k \leq \lg n - 1$ sum to at most $(3/2)\lg n$ and the terms for $k > \lg n - 1$ sum to at most $2$.
\end{proof}

%% file: text/extensions.tex
\section{Comments and Extensions}\label{sec:Remarks}

Zip trees combine two independent ideas: the use of random ranks distributed
geometrically and the use of unzipping and zipping to perform insertion and
deletion. The former saves space as compared to treaps and makes zip trees
isomorphic to skip lists but more efficient. In practice, allocating a byte
(eight bits) per rank should suffice. The latter makes updates faster as
compared to using rotations. Either idea may be used separately.

As compared to other kinds of search trees with logarithmic search time, zip
trees are simple and efficient: insertion and deletion can be done purely
top-down, with $O(1)$ expected restructuring time and exponentially infrequent
occurrences of expensive restructuring. Certain kinds of deterministic balanced
search trees, in particular weak AVL trees and red-black trees, achieve these
bounds in the amortized case \cite{Haeupler}, but at the cost of somewhat
complicated update algorithms.
Zipping and unzipping make catenating and splitting zip trees simple. To
catenate two zip trees $T_1$ and $T_2$ such that all items in $T_1$ have
smaller keys than those in $T_2$, zip the right spine of $T_1$ and the left
spine of $T_2$. The top node of the zipped path is the root of the new tree. To
split a tree into two, one containing items with keys at most $k$ and one
containing items with keys greater than $k$, unzip the path from the root down
to the node $x$ with key $k$, or down to a missing node if no item has key $k$.
The top nodes of the two unzipped paths are the roots of the new trees.

If the rank of a node is a pseudo-random function of its key, then search and
insertion can be combined into a single top-down operation that searches until
reaching the desired node or the insertion position. Similarly, search and
deletion can be so combined. Furthermore ranks need not be stored in nodes, but
can be computed as needed. However, for our efficiency analysis to hold, this
approach requires the stronger independence assumption that the sequence of
insertions and deletions is independent of the function generating the
ranks.\footnote{This issue is not merely theoretical. Reuse of random seeds has
led to real-world ``denial-of-service'' attacks for a number of programming
libraries. See \url{http://ocert.org/advisories/ocert-2011-003.html}.}
One additional nice feature of zip trees is that deletion does not require swapping a
binary node for a leaf or unary node before deleting it, as in Hibbard deletion \cite{Hibbard}.
As compared to treaps, zip trees have an average height about 8\% greater. By
choosing the ranks using a geometric distribution with higher mean, we can
reduce this discrepancy, at the cost of increasing the number of bits needed to
represent the ranks. Whether this is worthwhile is a question for experimental
study.

We believe that the properties of zip trees make them a good candidate for
concurrent implementation. The third author developed a preliminary, lock-based
implementation of concurrent zip trees in his senior thesis \cite{Timmel}. We
plan to develop a non-blocking implementation.

%% file: text/implementation.tex
\section{Implementations}\label{sec:Implementation}

In this section we present pseudocode implementing zip tree insertion and
deletion. We leave the implementation of search as an exercise. Our pseudocode
assumes an endogenous representation (nodes are items), with each node $x$
having a key $x.\Key$, a rank $x.\Rank$, and pointers to the left and right
children $x.\Left$ and $x.\Right$ of $x$ respectively. We give two
implementations designed to achieve different goals.

Our first goal is to minimize lines of code. This we do by using recursion. Our
recursive methods for insertion and deletion appear in Algorithm
\ref{impl:Recursive}. Method $\mathtt{insert}(x,\Root)$ inserts node $x$ into
the tree with root \emph{root} and returns the root of the resulting tree.  It requires that $x$ not be in the initial tree.
Once the last line of \texttt{insert} (``return \emph{root}'') in Algorithm \ref{impl:Recursive} is reached,
\texttt{insert} can actually return from the outermost call: all further tests
will fail, and no additional assignments will be done.
 Method
$\mathtt{delete}(x,\Root)$ deletes node $x$ from the tree with root \emph{root}
and returns the root of the resulting tree. It requires that $x$ be in the
initial tree. Unzipping is built into the insertion method; in deletion,
zipping is done by the separate method $\mathtt{zip}(x,y)$, which zips the
paths with top nodes $x$ and $y$ and returns the top node of the resulting
path. It requires that all descendants of $x$ have smaller key than all
descendants of $y$.

Our second goal is to do updates completely top-down and to minimize pointer
changes. This results in longer, less elegant, but more straightforward
methods. We treat \emph{root} as a global variable, with $\Root =
\mathtt{null}$ indicating an empty tree. Method $\mathtt{insert}(x)$ in
Algorithm \ref{list:IterativeInsert} inserts node $x$ into the tree with root
\emph{root}, assuming that $x$ is not already in the tree. Method
$\mathtt{delete}(x)$ in Algorithm \ref{list:IterativeDelete} deletes node $x$
from the tree with root \emph{root}, assuming it is in the tree.

These methods do some redundant tests and assignments to local variables. These
could be eliminated by loop unrolling, but might also be eliminated by a good
optimizing compiler.

\begin{algorithm}
  \DontPrintSemicolon
  \SetKwFunction{Insert}{insert}
  \SetKwFunction{Delete}{delete}
  \SetKwFunction{Zip}{zip}
  \BlankLine

  \Indm
  \nonl\hspace{0.5em}\Insert{$x$, $\Root$}:\;
  \Indp
  \lIf{$\Root=\Null$}{\{$x.\Left\gets x.\Right\gets\Null$; 
        $x.\Rank\gets\mathtt{RandomRank}$;
      return $x$\}}
  \eIf{$x.\Key<\Root.\Key$}
  {
    \If{$\text{\Insert{$x$, $\Root.\Left$}}=x$}
    {
      \lIf{$x.\Rank<\Root.\Rank$}{$\Root.\Left\gets x$}
      \lElse{\{$\Root.\Left\gets x.\Right$; $x.\Right\gets\Root$; return 
        $x$\}}
    }
  }
  {
    \If{$\text{\Insert{$x$, $\Root.\Right$}}=x$}
    {
      \lIf{$x.\Rank\le\Root.\Rank$}{$\Root.\Right\gets x$}
      \lElse{\{$\Root.\Right\gets x.\Left$; $x.\Left\gets\Root$; return 
        $x$\}}
    }
  }
  return $\Root$
  \BlankLine\BlankLine

  \Indm
  \nonl\hspace{0.5em}\Zip{$x$, $y$}:\;
  \Indp
  \lIf{$x=\Null$}{return $y$}
  \lIf{$y=\Null$}{return $x$}
  \lIf{$x.\Rank<y.\Rank$}{\{$y.\Left\gets\text{\Zip{$x$, $y.\Left$}}$; return $y$\}}
  \lElse{\{$x.\Right\gets\text{\Zip{$x.\Right$, $y$}}$; return $x$\}}
  \BlankLine\BlankLine

  \Indm
  \nonl\hspace{0.5em}\Delete{$x$, $\Root$}:\;
  \Indp
  \lIf{$x.\Key=\Root.\Key$}{return \Zip{$\Root.\Left$, $\Root.\Right$}}
  \eIf{$x.\Key<\Root.\Key$}
  {
    \If{$x.\Key=\Root.\Left.\Key$}
      {$\Root.\Left\gets\text{\Zip{$\Root.\Left.\Left$, $\Root.\Left.\Right$}}$}
    \lElse{\Delete{$x$, $\Root.\Left$}}
  }
  {
    \If{$x.\Key=\Root.\Right.\Key$}
    {$\Root.\Right\gets
      \text{\Zip{$\Root.\Right.\Left$, $\Root.\Right.\Right$}}$}
    \lElse{\Delete{$x$, $\Root.\Right$}}
  }
  return $\Root$
  \BlankLine

  \caption{Recursive versions of insertion and deletion.}
  \label{impl:Recursive}
\end{algorithm}

\newcommand\Prev{\mathit{prev}}
\newcommand\Cur{\mathit{cur}}
\newcommand\Fix{\mathit{fix}}
\newcommand\Next{\mathit{next}}

\begin{algorithm}
  \DontPrintSemicolon
  \BlankLine\BlankLine
  \Indm
  \nonl\hspace{0.5em}\Insert{$x$}\;
  \Indp
  $\Rank\gets x.\Rank\gets\mathtt{RandomRank}$\;
  $\Key\gets x.\Key$\;
  $\Cur\gets \Root$\;
  \While{$\Cur\ne\mathtt{null}$\emph{ and (}$\Rank<\Cur.\Rank$\emph{ or
    (}$\Rank=\Cur.\Rank$ \emph{and} $\Key>\Cur.\Key$\emph{))}}
    {
    $\Prev\gets\Cur$\;
    $\Cur\gets\text{if $\Key<\Cur.\Key$ then $\Cur.\Left$ else $\Cur.\Right$}$\;
    }
  \BlankLine

  \lIf{$\Cur=\Root$}{$\Root\gets x$}
  \lElseIf{$\Key<\Prev.\Key$}{$\Prev.\Left\gets x$}
  \lElse{$\Prev.\Right\gets x$}
  \BlankLine

  \lIf{$\Cur=\mathtt{null}$}{\{$x.\Left\gets
    x.\Right\gets\mathtt{null}$; return\}}
  \leIf{$\Key<\Cur.\Key$}{$x.\Right\gets\Cur$}{$x.\Left\gets\Cur$}
  $\Prev\gets x$\;
  \BlankLine

  \While{$\Cur\ne\mathtt{null}$}
  {
    $\Fix\gets\Prev$\;
    \eIf{$\Cur.\Key<\Key$}
      {
        \textbf{repeat} \{$\Prev\gets\Cur$; $\Cur\gets\Cur.\Right$\}\;
        \textbf{until} $\Cur=\mathtt{null}$ or $\Cur.\Key>\Key$\;
      }
      {
        \textbf{repeat} \{$\Prev\gets\Cur$; $\Cur\gets\Cur.\Left$\}\;
        \textbf{until} $\Cur=\mathtt{null}$ or $\Cur.\Key<\Key$\;
      }
    \eIf{$\Fix.\Key>\Key$\emph{ or (}$\Fix
      =x$\emph{ and }$\Prev.\Key>\Key$\emph{)}}
    {$\Fix.\Left\gets\Cur$}{$\Fix.\Right\gets\Cur$}
  }
  \BlankLine\BlankLine\BlankLine
  \caption{Iterative insertion.}
  \label{list:IterativeInsert}
\end{algorithm}

\begin{algorithm}
  \DontPrintSemicolon
  \BlankLine\BlankLine
  \Indm
  \nonl\hspace{0.5em}\Delete{$x$}\;
  \Indp
  $\Key\gets x.\Key$\;
  $\Cur\gets\Root$\;
  \While{$\Key\ne\Cur.\Key$}
  {
    $\Prev\gets\Cur$\;
    $\Cur\gets\text{if $\Key<\Cur.\Key$ then $\Cur.\Left$ else $\Cur.\Right$}$\;
  }
  $\Left\gets\Cur.\Left$; $\Right\gets\Cur.\Right$\;
  \BlankLine

  \lIf{$\Left=\mathtt{null}$}{$\Cur\gets\Right$}
  \lElseIf{$\Right=\mathtt{null}$}{$\Cur\gets\Left$}
  \lElseIf{$\Left.\Rank\ge\Right.\Rank$}{$\Cur\gets\Left$}
  \lElse{$\Cur\gets\Right$}
  \BlankLine

  \lIf{$\Root=x$}{$\Root\gets\Cur$}
  \lElseIf{$\Key<\Prev.\Key$}{$\Prev.\Left\gets\Cur$}
  \lElse{$\Prev.\Right\gets\Cur$}
  \BlankLine

  \While{$\Left\ne\mathtt{null}$\emph{ and }$\Right\ne\mathtt{null}$}
  {
    \eIf{$\Left.\Rank\ge\Right.\Rank$}
    {
      \textbf{repeat} \{$\Prev\gets\Left$; $\Left\gets\Left.\Right$\}\;
      \textbf{until} $\Left=\mathtt{null}$ or $\Left.\Rank<\Right.\Rank$\;
      $\Prev.\Right\gets\Right$\;
    }
    {
      \textbf{repeat} \{$\Prev\gets\Right$; $\Right\gets\Right.\Left$\}\;
      \textbf{until} $\Right=\mathtt{null}$ or $\Left.\Rank\ge\Right.\Rank$\;
      $\Prev.\Left\gets\Left$
    }
  }
  \BlankLine\BlankLine\BlankLine
  \caption{Iterative deletion.}
  \label{list:IterativeDelete}
\end{algorithm}

\clearpage

%% file: text/acknowledgements.tex
\begin{acks}
We thank Dave Long for carefully reading the manuscript and offering many
useful suggestions, most importantly helping us simplify the iterative
insertion and deletion algorithms. We thank Sebastian Wild for correcting the
bound on expected node depth in treaps in Section \ref{sec:PreviousWork} and
for his ideas on breaking rank ties. We thank Leon Sering for correcting the
statement and proof of Lemma \ref{lem:LowAncestors}. Finally, we are grateful
to Dominik Kempa for providing us with C++ zip-tree implementations,
benchmarks, and general comments. Research at Princeton University by the first
two authors was partially supported by an innovation research grant from
Princeton and a gift from Microsoft.
\end{acks}